\documentclass[conference]{IEEEtran}
\usepackage{amsmath,amssymb,amsthm,cite,graphicx,array}
\usepackage{algorithm}
\usepackage[noend]{algpseudocode}
\theoremstyle{definition}
\newtheorem{theorem}{Theorem}
\newtheorem{lemma}{Lemma}
\newtheorem{corollary}{Corollary}

\newtheorem{definition}{Definition}
\newtheorem{remark}{Remark}
\newtheorem{example}{Example}
\usepackage{color}

\title{Near-Optimal Vector Linear Index Codes For Single Unicast Index Coding Problems with Symmetric Neighboring Interference}

\begin{document}

\author{Mahesh~Babu~Vaddi~and~B.~Sundar~Rajan\\ 
 Department of Electrical Communication Engineering, Indian Institute of Science, Bengaluru 560012, KA, India \\ E-mail:~\{mahesh,~bsrajan\}@ece.iisc.ernet.in }
 
\maketitle
\begin{abstract}
A single unicast index coding problem (SUICP) with symmetric neighboring interference (SNI) has equal number of $K$ messages and $K$ receivers, the $k$th receiver $R_{k}$ wanting the $k$th message $x_{k}$ and having the side-information $\mathcal{K}_{k}=(\mathcal{I}_{k} \cup x_{k})^c,$ where ${I}_k= \{x_{k-U},\dots,x_{k-2},x_{k-1}\}\cup\{x_{k+1}, x_{k+2},\dots,x_{k+D}\}$ is the interference with $D$ messages after and $U$ messages before its desired message. Maleki, Cadambe and Jafar obtained the capacity of this single unicast index coding problem with symmetric neighboring interference (SUICP-SNI) with $K$ tending to infinity and Blasiak, Kleinberg and Lubetzky for the special case of $(D=U=1)$ with $K$ being finite. In our previous work, we proved the capacity of SUICP-SNI for arbitrary $K$ and $D$ with $U=\text{gcd}(K,D+1)-1$. This paper deals with near-optimal linear code construction for SUICP-SNI  with arbitrary $K,U$ and $D.$  For SUICP-SNI with arbitrary $K,U$ and $D$, we define a set of $2$-tuples such that for every $(a,b)$ in that set the rate $D+1+\frac{a}{b}$ is achieved by using vector linear index codes over every field. We prove that the set $\mathcal{\mathbf{S}}$ consists of $(a,b)$ such that the rate of constructed vector linear index codes are at most $\frac{K~\text{mod}~(D+1)}{\left \lfloor \frac{K}{D+1} \right \rfloor}$ away from a known lower bound on broadcast rate of SUICP-SNI. The three known results on the exact capacity of the SUICP-SNI are recovered as special cases of our results. Also, we give a low complexity decoding procedure for the proposed vector linear index codes for the SUICP-SNI.
\end{abstract}
\section{Introduction and Background}
\label{sec1}
\IEEEPARstart {A}{n} index coding problem, comprises a transmitter that has a set of $K$ independent messages, $X=\{ x_0,x_1,\ldots,x_{K-1}\}$, and a set of $M$ receivers, $R=\{ R_0,R_1,\ldots,R_{M-1}\}$. Each receiver, $R_k=(\mathcal{K}_k,\mathcal{W}_k)$, knows a subset of messages, $\mathcal{K}_k \subset X$, called its \textit{Known-set} or the \textit{side-information}, and demands to know another subset of messages, $\mathcal{W}_k \subseteq \mathcal{K}_k^\mathsf{c}$, called its \textit{Want-set} or \textit{Demand-set}. The transmitter can take cognizance of the side-information of the receivers and broadcast coded messages, called the index code, over a noiseless channel. The objective is to minimize the number of coded transmissions, called the length of the index code, such that each receiver can decode its demanded message using its side-information and the coded messages.

The problem of index coding with side-information was introduced by Birk and Kol \cite{ISCO}. Ong and Ho \cite{OnH} classified the binary index coding problem depending on the demands and the side-information possessed by the receivers. An index coding problem is unicast if the demand-sets of the receivers are disjoint. An index coding problem is single unicast if the demand-sets of the receivers are disjoint and the cardinality of demand-set of every receiver is one. Any unicast index coding problem can be converted into a single unicast index coding problem. A single unicast index coding problem (SUICP) can be described as follows: Let $\{x_{0}$,$x_{1}$,\ldots,$x_{K-1}\}$ be the $K$ messages, $\{R_{0}$,$R_{1},\ldots,R_{K-1}\}$ are $K$ receivers and $x_k \in \mathcal{A}$ for some alphabet $\mathcal{A}$ and $k=0,1,\ldots,K-1$. Receiver $R_{k}$ wants the message $x_{k}$ and knows a subset of messages in $\{x_{0}$,$x_{1}$,\ldots,$x_{K-1}\}$ as side-information. 


A solution (includes both linear and nonlinear) of the index coding problem must specify a finite alphabet $\mathcal{A}_P$ to be used by the transmitter, and an encoding scheme $\varepsilon:\mathcal{A}^{t} \rightarrow \mathcal{A}_{P}$ such that every receiver is able to decode the wanted message from $\varepsilon(x_0,x_1,\ldots,x_{K-1})$ and the known information. The minimum encoding length $l=\lceil log_{2}|\mathcal{A}_{P}|\rceil$ for messages that are $t$ bit long ($\vert\mathcal{A}\vert=2^t$) is denoted by $\beta_{t}(G)$. The broadcast rate of the index coding problem with side-information graph $G$ is defined \cite{ICVLP} as,
$\beta(G) \triangleq   \inf_{t} \frac{\beta_{t}(G)}{t}.$
If $t = 1$, it is called scalar broadcast rate. For a given index coding problem, the broadcast rate $\beta(G)$ is the minimum number of index code symbols required to transmit to satisfy the demands of all the receivers. The capacity $C(G)$ for the index coding problem is defined as the maximum number of message symbols transmitted per index code symbol such that every receiver gets its wanted message symbols and all the receivers get equal number of wanted message symbols. The broadcast rate and capacity are related as 
\begin{center}	
$C(G)=\dfrac{1}{\beta(G)}$.
\end{center}  

Instead of one transmitter and $K$ receivers, the SUICP-SNI can also be viewed as $K$ source-receiver pairs with all $K$ sources connected with all $K$ receivers through a common finite capacity channel and all source-receiver pairs connected with either zero of infinite capacity channels. This problem is called multiple unicast index coding problem in \cite{MCJ}.

\subsection{Single Unicast Index Coding Problem with symmetric Neighboring Interference}
A single unicast index coding problem with symmetric neighboring interference (SUICP-SNI) with equal number of $K$ messages and receivers, is one with each receiver having a total of $U+D<K$ interference, corresponding to the $D~(U \leq D)$ messages above and $U$ messages before its desired message. In this setting, the $k$th receiver $R_{k}$ demands the message $x_{k}$ having the interference
\begin{equation}
\label{antidote}
{I}_k= \{x_{k-U},\dots,x_{k-2},x_{k-1}\}\cup\{x_{k+1}, x_{k+2},\dots,x_{k+D}\}, 
\end{equation}
\noindent
the side-information being 
$\mathcal{K}_{k}=(\mathcal{I}_{k} \cup x_{k})^c.$

Maleki \textit{et al.} \cite{MCJ} found the capacity of SUICP-SNI with $K\rightarrow \infty$ to be  
\begin{align}
\label{cap1}
C=\frac{1}{D+1},
\end{align}
and an upper bound for the capacity of SUICP-SNI for finite $K$ to be 
\begin{align}
\label{outerbound}
C \leq \frac{1}{D+1},
\end{align}
which is same as the broadcast rate of the SUICP-SNI being lower bounded as
\begin{align}
\label{cap4}
\beta \geq D+1.
\end{align}

Blasiak \textit{et al.} \cite{ICVLP} found the capacity of SUICP-SNI with $U=D=1$ by using linear programming bounds to be 
\begin{align}
\label{cap2}
C=\frac{\left\lfloor \frac{K}{2}\right\rfloor}{K}. 
\end{align}

In \cite{VaR4}, we showed that the capacity of SUICP-SNI for arbitrary $K$ and $D$ with $U=\text{gcd}(K,D+1)-1$ is 
\begin{align}
\label{cap3}
C=\frac{1}{D+1}.
\end{align}

In \cite{VaR2}, we constructed binary matrices of size $m \times n (m\geq n)$ such that any $n$ adjacent rows of the matrix are linearly independent over every field. Calling these matrices as Adjacent Independent Row (AIR) matrices, we gave an optimal scalar linear index code for the single unicast index coding problems with symmetric neighboring consecutive side information (SUICP-SNC)(one-sided)  using AIR encoding matrices. In \cite{VaR1}, we constructed optimal vector linear index codes for  SUICP-SNC (two-sided) In \cite{VaR3}, we gave a low-complexity decoding for SUICP-SNC with AIR matrix as encoding matrix. The low complexity decoding method helps to identify a reduced set of side-information for each user with which the decoding can be carried out. By this method every receiver is able to decode its wanted message symbol by simply adding some index code symbols (broadcast symbols).

\subsection{Contributions}
In this paper we deal with arbitrary $K,D$ and $U$ for SUICP-SNI. The contributions of this paper are summarized below:
\begin{itemize}
\item For SUICP-SNI with arbitrary $K,U$ and $D$, we define a set $\mathcal{\mathbf{S}}$ of $2$-tuples such that for every $(a,b) \in \mathcal{\mathbf{S}}$, the rate $D+1+\frac{a}{b}$ is achievable by using AIR matrices with vector linear index codes over every field.
\item We prove that the rates achieved by AIR matrices coincide with the existing results on the exact value of the capacity of SUICP-SNI given in \eqref{cap1},\eqref{cap2} and \eqref{cap3} in the respective settings. 
\item We prove that the set $\mathcal{\mathbf{S}}$ consists of $(a,b)$ such that the rate of constructed vector linear index codes are at most $\frac{K~\text{mod}~(D+1)}{\left \lfloor \frac{K}{D+1} \right \rfloor}$ away from a known lower bound on broadcast rate of SUICP-SNI.
\item We give a low complexity decoding for the proposed vector linear index codes for the SUICP-SNI. The low complexity decoding method helps to identify a reduced set of side-information for each user with which the decoding can be carried out. By this method every receiver is able to decode its wanted message symbol by simply adding some index code symbols (broadcast symbols).

\end{itemize}

For a subset $I=\{i_1,i_2,\ldots,i_l\} \subseteq \{1,2,\ldots,K\}$, let $x_I=\{x_{i_1},x_{i_2},\ldots,x_{i_l}\}$ and $L_I=\{L_{i_1},L_{i_2},\ldots,L_{i_l}\}$. For vector subspaces $\mathbf{S}_{i_1},\mathbf{S}_{i_2},\ldots,\mathbf{S}_{i_l}$, let $\mathbf{S}_I=\{\mathbf{S}_{i_1},\mathbf{S}_{i_2},\ldots,\mathbf{S}_{i_l}\}$ and $<\mathbf{S}_I>=<\mathbf{S}_{i_1},\mathbf{S}_{i_2},\ldots,\mathbf{S}_{i_l}>$ denote the vector subspace spanned by all vectors present in $\mathbf{S}_{i_1},\mathbf{S}_{i_2},\ldots,\mathbf{S}_{i_l}$. For matrices $\mathbf{\tilde{S}}_{i_1},\mathbf{\tilde{S}}_{i_2},\ldots,\mathbf{\tilde{S}}_{i_l}$, let $\mathbf{\tilde{S}}_I=\{\mathbf{\tilde{S}}_{i_1},\mathbf{\tilde{S}}_{i_2},\ldots,\mathbf{\tilde{S}}_{i_l}\}$. 

Henceforth,  we refer SUICP-SNI with $K$ messages, $D$ interfering messages after and $U$ interfering messages before the desired message as $(K,D,U)$ SUICP-SNI.

The remaining part of this paper is organized as follows. In Section \ref{sec2}, for $(K,D,U)$ SUICP-SNI, we define a set $\mathbf{S}$ of $2$-tuples such that for every $(a,b) \in \mathcal{\mathbf{S}}$, the rate $D+1+\frac{a}{b}$ is achievable by using AIR matrices with vector linear index codes over every field (Theorem \ref{thm1}).  Also we prove that the rates achieved by AIR Matrices coincide with the existing results on the exact value of the capacity of SUICP-SNI given in \eqref{cap1},\eqref{cap2} and \eqref{cap3} in Theorem \ref{thm2}. Moreover, it is shown that the set $\mathcal{\mathbf{S}}$ consists of $(a,b)$ such that the rate of constructed vector linear index codes are at most $\frac{K~\text{mod}~(D+1)}{\left \lfloor \frac{K}{D+1} \right \rfloor}$ away from a known lower bound on broadcast rate of SUICP-SNI. In Section \ref{sec3}, we give a low complexity decoding for the proposed vector linear index codes. We conclude the paper in Section \ref{sec4}.

All the subscripts in this paper are to be considered $~\text{\textit{modulo}}~ K$. 
\section{Vector Linear Index Codes of SUICP-SNI: Achievability Results}
\label{sec2}
In this section, we define a set $\mathcal{\mathbf{S}}$ consisting of pairs of integers $(a,b)$ and prove that the rate $D+1+\frac{a}{b}$ for every $(a,b) \in \mathcal{\mathbf{S}}$ is achievable by using an appropriate sized AIR matrix as the encoding matrix.

\begin{definition}
\label{def1}
Consider the SUICP-SNI with $K$ messages, $D$ and $U$ interfering messages after and before  the desired message. For this SUICP-SNI, define the set $\mathbf{S}_{K,D,U}$ as  
\begin{align}
\label{ab}
\mathbf{S}_{K,D,U}=\{(a,b):\text{gcd}(bK,b(D+1)+a)\geq b(U+1)\}
\end{align}
for $a \in Z_{\geq 0}$ and $b \in Z_{>0}$. 
\end{definition}
\subsection{AIR matrices}
The matrix obtained by Algorithm \ref{algo2} is called the $(m,n)$ AIR matrix and it is denoted by $\mathbf{L}_{m\times n}.$ The general form of the $(m,n)$ AIR matrix is shown in   Fig. \ref{fig1}. It consists of several submatrices (rectangular boxes) of different sizes as shown in Fig.\ref{fig1}. The location and sizes of these submatrices and other quantities marked in the figure are used subsequently in Section \ref{sec3} to describe the low complexity decoding for $(K,D,U)$ SUICP-SNI by using AIR matrices.
\begin{figure*}[ht]
\centering
\includegraphics[scale=0.6]{matrix213}\\
~ $\mathbf{S}=\mathbf{I}_{\lambda_{l} \times \beta_l \lambda_{l}}$ if $l$ is even and ~$\mathbf{S}=\mathbf{I}_{\beta_l\lambda_{l} \times \lambda_{l}}$ otherwise.
\caption{AIR matrix of size $m \times n$.}
\label{fig1}
~ \\
\hrule
\end{figure*}
The description of the submatrices are as follows: Let $c$ and $d$ be two positive integers and $d$ divides $c$. The following matrix  denoted by $\mathbf{I}_{c \times d}$ is a rectangular  matrix.
\begin{align}
\label{rcmatrix}
\mathbf{I}_{c \times d}=\left.\left[\begin{array}{*{20}c}
   \mathbf{I}_{d}  \\
   \mathbf{I}_{d}  \\
   \vdots  \\
   \mathbf{I}_{d} 
   \end{array}\right]\right\rbrace \frac{c}{d}~\text{number~of}~ \mathbf{I}_{d}~\text{matrices}
\end{align}
and $\mathbf{I}_{d \times c}$ is the transpose of $\mathbf{I}_{c \times d}.$

Towards explaining the other quantities shown in the AIR matrix shown in Fig. \ref{fig1}, for a given $K,D$  and $U,$ let  $\lambda_{-1}=n,\lambda_0=m-n$ and\begin{align}
\nonumber
n&=\beta_0 \lambda_0+\lambda_1, \nonumber \\
\lambda_0&=\beta_1\lambda_1+\lambda_2, \nonumber \\
\lambda_1&=\beta_2\lambda_2+\lambda_3, \nonumber \\
\lambda_2&=\beta_3\lambda_3+\lambda_4, \nonumber \\
&~~~~~~\vdots \nonumber \\
\lambda_i&=\beta_{i+1}\lambda_{i+1}+\lambda_{i+2}, \nonumber \\ 
&~~~~~~\vdots \nonumber \\ 
\lambda_{l-1}&=\beta_l\lambda_l.
\label{chain}
\end{align}
where $\lambda_{l+1}=0$ for some integer $l,$ $\lambda_i,\beta_i$ are positive integers and $\lambda_i < \lambda_{i-1}$ for $i=1,2,\ldots,l$. The number of submatrices in the AIR matrix is $l+2$ and the size of each submatrix is shown using $\lambda_i,\beta_i,$  $i \in [0:l].$
For $(K,D,U)$ SUICP-SNI and every $(a,b)\in \mathbf{S}_{K,D,U}$, in Theorem \ref{thm1}, we prove that the AIR matrix of size $Kb \times (b(D+1)+a)$ can be used as an encoding matrix to generate $b$-dimensional vector linear index code with rate $D+1+\frac{a}{b}$.

		\begin{algorithm}
		{Algorithm to construct the AIR matrix $\mathbf{L}$ of size $m \times n$}
			\begin{algorithmic}[2]
				 \item Let $\mathbf{L}=m \times n$ blank unfilled matrix.
				\item [Step 1]~~~
				\begin{itemize}
				\item[\footnotesize{1.1:}] Let $m=qn+r$ for $r < n$.
				\item[\footnotesize{1.2:}] Use $\mathbf{I}_{qn \times n}$ to fill the first $qn$ rows of the unfilled part of $\mathbf{L}$.
				\item[\footnotesize{1.3:}] If $r=0$,  Go to Step 3.
				\end{itemize}

				\item [Step 2]~~~
				\begin{itemize}
				\item[\footnotesize{2.1:}] Let $n=q^{\prime}r+r^{\prime}$ for $r^{\prime} < r$.
				\item[\footnotesize{2.2:}] Use $\mathbf{I}_{q^{\prime}r \times r}^{\mathsf{T}}$ to fill the first $q^{\prime}r$ columns of the unfilled part of $\mathbf{L}$.
			    \item[\footnotesize{2.3:}] If $r^{\prime}=0$,  go to Step 3.	
				\item[\footnotesize{2.4:}] $m\leftarrow r$ and $n\leftarrow r^{\prime}$.
				\item[\footnotesize{2.5:}] Go to Step 1.
				\end{itemize}
				\item [Step 3] Exit.
		
			\end{algorithmic}
			\label{algo2}
		\end{algorithm}

In an $m$-dimensional vector linear index code, $x_k \in \mathbb{F}^m_q$ for $k \in [0:K-1]$. A $m$-dimensional vector linear index code of length $N$ is represented by an encoding matrix $\mathbf{L}$ $(\in \mathbb{F}^{Km\times N}_q)$, where the $j$th column contains the coefficients used for mixing the $m$-dimensional messages $x_0,x_1,\ldots,x_{K-1}$ to get the $j$th index code symbol. Let $L_0,L_1,\ldots, L_{Km-1}$ be the $Km$ rows of the encoding matrix $\mathbf{L}$. Let $\mathbf{\tilde{S}}_k$ be the $m \times N$ matrix 
$$\mathbf{\tilde{S}}_{k}=\left.\left[\begin{array}{*{20}c}
   L_{km}  \\
   L_{km+1}  \\
   \vdots  \\
  L_{km+m-1}   
   \end{array}\right]\right.$$
for $k \in [0:K-1]$. The $k$th matrix $\mathbf{\tilde{S}}_k$ $(\in \mathbb{F}^{m\times N}_q)$ contains the coefficients used for mixing the $m$-dimensional message $x_k$ in the $N$ index code symbols. A codeword of the index code is 
\begin{align*}
\mathbf{y}=[c_0~c_1~\ldots~c_{N-1}]=\mathbf{xL}=\sum_{i=0}^{K-1}x_i\mathbf{\tilde{S}}_k,
\end{align*}
where $\mathbf{x}=[x_{0,1}~x_{0,2}~\ldots~x_{0,m}~x_{1,1}~x_{1,2}~\ldots~x_{1,m}~\ldots~x_{K-1,m}]$.
\begin{lemma}
\label{lemma1}
Consider a $(K,D,U)$ SUICP-SNI. Let $\mathbf{L}$ be an $m$-dimensional vector encoding matrix (not necessarily an AIR matrix) of size $Km\times N$ for this index coding problem. Let the vector index code be generated by multiplying $Km$ message symbols $[x_{0,1}~x_{0,2}~\ldots~x_{0,m}~x_{1,1}~x_{1,2}~\ldots~x_{1,m}~\ldots~x_{K-1,m}]$ with the encoding matrix $\mathbf{L}$. Let $\mathbf{S}_k$ be the row space of $\tilde{\mathbf{S}}_k$ and $\mathbf{S}_{k \setminus i}$ be the subspace spanned by the $m-1$ rows in $\tilde{\mathbf{S}}_k$ after deleting the row $L_{km+i}$. The receiver $R_k$ can decode $x_{k,1},x_{k,2},\ldots,x_{k,m}$ if and only if 
\begin{align} 
\label{sind}
L_{km+i} \notin  < \mathbf{S}_{{\cal{I}}_k},\mathbf{S}_{k \setminus i} >
\end{align}
for $k\in [0:K-1]$, $i \in [0:m-1]$.
\end{lemma}
\begin{proof}
The $m$-dimensional vector index codeword $\mathbf{y}$ can be written as
\begin{align}
\nonumber
\mathbf{y}&=x_{{\cal{K}}_k}\mathbf{\tilde{S}}_{{\cal{K}}_k}+x_k\mathbf{\tilde{S}}_k+x_{{\cal{I}}_k}\mathbf{\tilde{S}}_{{\cal{I}}_k}, \text{ or}
\\ \mathbf{z}&=\mathbf{y}-x_{{\cal{K}}_k}\mathbf{\tilde{S}}_{{\cal{K}}_k}=x_k\mathbf{\tilde{S}}_k+x_{\mathcal{I}_k}\mathbf{\tilde{S}}_{\mathcal{I}_k},
\label{sdecoding}
\end{align}
where $\mathbf{z}$ can be computed by $R_k$ using its side-information $x_{{{\cal{K}}_k}}$.

Assume that \eqref{sind} is satisfied for $k$. This implies that
\begin{align*}
\mathbf{S}_k \cap <\mathbf{S}_{k-U},\mathbf{S}_{k-U+1},\ldots,\mathbf{S}_{k-1},\mathbf{S}_{k+1},\ldots,\mathbf{S}_{k+D}>=\phi.
\end{align*}

Hence, $\mathbf{z}$ can be expressed as the following linear combination
\begin{align*}
\mathbf{z}=&a_k\mathbf{\tilde{S}}_{k}+a_{1}\mathbf{\tilde{S}}_{k-1}+a_{2}\mathbf{\tilde{S}}_{k-2}+\ldots+a_{k-D}\mathbf{\tilde{S}}_{k-D} 
\\&+a_{k+1}\mathbf{\tilde{S}}_{k+1}+a_{k+2}\mathbf{\tilde{S}}_{k+2}+\ldots+a_{k+D}\mathbf{\tilde{S}}_{k+D},
\end{align*}
where $a_k \in \mathbb{F}^m_q$ is unique and $x_k=a_k$. The receiver $R_k$ decodes all $m$ message symbols in the message $x_k$ follows from the fact $L_{(k-1)m+i} \notin  < \mathbf{S}_{k \setminus i} >$.

If, on the contrary, $\mathbf{S}_{k}$ does not satisfy \eqref{sind}, then $a_{k}$ will no longer be unique and consequently $x_{k}$ can not be decoded by $R_k$. This implies that $\mathbf{L}$ is not an index code encoding matrix which contradicts the assumption. This completes the proof.
\end{proof}
\begin{corollary}
\label{cor1}
Consider a $(K,D,U)$ SUICP-SNI. Let $\mathbf{L}$ be an scalar linear encoding matrix (not necessarily an AIR matrix) of size $K \times N$ for this index coding problem. The receiver $R_k$ can decode $x_k$ if and only if 
\begin{align} 
\label{sind2}
L_{k} \notin  < L_{{\cal{I}}_k}>
\end{align}
for $k\in [0:K-1]$.
\end{corollary}
\begin{lemma}
\label{lemma2}
Let $\mathbf{L}$ be the AIR matrix of size $K \times (D+1)$. The matrix $\mathbf{L}$ can be used as an optimal length encoding matrix for the SUICP-SNI with $K$ messages, $D$ interfering messages after and $U=\text{gcd}(K,D+1)-1$ interfering messages before the desired message.
\end{lemma}
\begin{proof}
Proof is given in Appendix B of \cite{VaR4}.
\end{proof}
\begin{lemma}
\label{lemma3}
In the AIR matrix of size $m \times n$, every row $L_k$ is not in the span of $n-1$ rows above and $\text{gcd}(m,n)-1$ rows below $L_k$ for $k \in [0:m-1]$.
\end{lemma}
\begin{proof}
Consider a SUICP-SNI with $m$ messages, $n-1$ interfering messages after and $\text{gcd}(m,n)-1$ interfering messages before the desired message. From Lemma \ref{lemma2},  AIR matrix of size $m \times n$ can be used as an optimal length encoding matrix for this SUICP-SNI. Let $\mathbf{L}$ be the AIR matrix of size $m \times n$. From Corollary \ref{cor1}, $\mathbf{L}$ is an scalar linear encoding matrix for SUICP-SNI if and only if 
\begin{align*}
L_k \notin &< L_{{\cal{I}}_k}>\\&=<L_{k-\text{gcd}(m,n)+1},L_{k-\text{gcd}(m,n)+2},\ldots,L_{k-1},\\&~~~~~~~~~~~~~~~~~~~~~~~~~L_{k+1},L_{k+2},\ldots,L_{k+n-1}>. 
\end{align*}

That is, every row $L_k$ in $\mathbf{L}$ for $k \in [0:m-1]$ is not in the span of $n$ rows above and $\text{gcd}(m,n)-1$ rows below the $L_k$. This completes the proof.
\end{proof}
\begin{theorem}
\label{thm1}
Consider a $(K,D,U)$ SUICP-SNI. For this index coding problem, for every $(a,b) \in \mathbf{S}_{K,D,U}$, the rate $D+1+\frac{a}{b}$ can be achieved by $b$-dimensional vector linear index coding by using the AIR matrix of size $Kb \times (b(D+1)+a)$. 
\end{theorem}
\begin{proof}
In a $b$-dimensional vector linear index code, $x_k \in \mathbb{F}^{b}_q$ for $k \in [0:K-1]$, the receiver $R_k$ wants the $b$ message symbols $x_{k,1}~x_{k,2}\ldots x_{k,b}$. Let 
\begin{align*}
\mathbf{x}=[\underbrace{x_{0,1}x_{0,2}\ldots x_{0,b}}_{x_0}~\underbrace{x_{1,1}x_{1,2}\ldots x_{1,b}}_{x_1}\ldots \underbrace{x_{K-1,1}\ldots x_{K-1,b}}_{x_{K-1}}].
\end{align*}

Let $\mathbf{L}$ be the AIR matrix  of size $Kb \times (b(D+1)+a)$. From Lemma \ref{lemma3}, every row $L_k$ in $\mathbf{L}$ is not in the span of $b(D+1)+a-1$ rows above and $\text{gcd}(Kb,b(D+1)+a)-1$ rows below the $L_k$ for each $k \in [0:Kb-1]$. From Definition \ref{def1}, $a$ and $b$ are the positive integers satisfying the relation
\begin{align*}
\text{gcd}(Kb,b(D+1)+a) \geq b(U+1).
\end{align*}

Hence, every row $L_k$ in $\mathbf{L}$ is not in the span of $b(D+1)+a-1$ rows above and $b(U+1)-1$ rows below the $L_k$ for each $k \in [0:Kb-1]$. According to Lemma \ref{lemma1},  the matrix $\mathbf{L}$ can be used as a $b$-dimensional encoding matrix for SUICP-SNI with $K$ messages, $D$ and $U$ interfering messages after and before.

The matrix $\mathbf{L}$ is mapping $Kb$ message symbols into $b(D+1)+a$ broadcast symbols. The rate achieved by $\mathbf{L}$ is equal to the number of broadcast symbols per message symbol and it is given by
$\frac{b(D+1)+a}{b}=D+1+\frac{a}{b}.$ 
 This completes the proof.
\end{proof}
\begin{remark}
The AIR encoding matrix $\mathbf{L}_{Kb \times (b(D+1)+a)}$ is an encoding matrix over every field. Hence, the encoding for $(K,D,U)$ SUICP-SNI given in Theorem \ref{thm1} is independent of field size.
\end{remark}
\begin{example}
\label{ex10}
Consider the SUICP-SNI with $K=13,D=4,U=1$. For this case, we have $(1,5) \in \mathbf{S}_{13,4,1}$ and corresponding  $D+1+\frac{a}{b}$ is $5.2$. This rate can be achieved by the AIR matrix of size $65 \times 26$ given below.
$$\mathbf{L}_{65 \times 26}=\left.\left[\begin{array}{*{20}c}
   \mathbf{I}_{26}  \\
   \mathbf{I}_{26}  \\
   \mathbf{I}_{13}~\mathbf{I}_{13} \\   
   \end{array}\right]\right.$$
   The index code generated by $\mathbf{L}$ is given by 
   \begin{align*}
   &[c_0~c_1~\ldots c_{25}]=\mathbf{x}_{1 \times 65}\mathbf{L}_{65 \times 26}
   \end{align*}
   where $\mathbf{x}_{1 \times 65}=[x_{0,1}~x_{0,2}\ldots x_{0,5}~x_{1,1}~x_{1,2}\ldots x_{1,5}\ldots x_{12,5}]$.
   The index code is given in Table \ref{table11}. Receiver $R_k$ required to decode five message symbols $x_{k,1},x_{k,2},x_{k,3},x_{k,4}$ and $x_{k,5}$ for $k \in [0:12]$. Let $\tau_{k,j}$ be the code symbols used by receiver $R_k$ to decode $x_{k,j}$ for $k \in [0:12]$ and $j \in [1:5]$. Table \ref{table12} gives the code symbols used by each receiver to decode its wanted message symbol. The details of the decoding procedure is given in Section \ref{sec3}.
   \begin{table}[ht]
\centering
\setlength\extrarowheight{0pt}
\begin{tabular}{|c|c|}
\hline
$c_0=x_{0,1}+x_{5,2}+x_{10,3}$ & $c_{13}=x_{2,4}+x_{7,5}+x_{10,3}$\\
\hline
$c_1=x_{0,2}+x_{5,3}+x_{10,4}$ & $c_{14}=x_{2,5}+x_{8,1}+x_{10,4}$ \\
\hline
$c_2=x_{0,3}+x_{5,4}+x_{10,5}$ & $c_{15}=x_{3,1}+x_{8,2}+x_{10,5}$ \\
\hline
$c_3=x_{0,4}+x_{5,5}+x_{11,1}$ & $c_{16}=x_{3,2}+x_{8,3}+x_{11,1}$ \\
\hline
$c_4=x_{0,5}+x_{6,1}+x_{11,2}$ &  $c_{17}=x_{3,3}+x_{8,4}+x_{11,2}$ \\
\hline
$c_5=x_{1,1}+x_{6,2}+x_{11,3}$ & $c_{18}=x_{3,4}+x_{8,5}+x_{11,3}$ \\
\hline
$c_6=x_{1,2}+x_{6,3}+x_{11,4}$ & $c_{19}=x_{3,5}+x_{9,1}+x_{11,4}$ \\
\hline
$c_7=x_{1,3}+x_{6,4}+x_{11,5}$ & $c_{20}=x_{4,1}+x_{9,2}+x_{11,5}$\\
\hline
$c_8=x_{1,4}+x_{6,5}+x_{12,1}$ & $c_{21}=x_{4,2}+x_{9,3}+x_{12,1}$ \\
\hline
$c_9=x_{1,5}+x_{7,1}+x_{12,2}$& $c_{22}=x_{4,3}+x_{9,4}+x_{12,2}$ \\
\hline
$c_{10}=x_{2,1}+x_{7,2}+x_{12,3}$ & $c_{23}=x_{4,4}+x_{9,5}+x_{12,3}$ \\
\hline
$ c_{11}=x_{2,2}+x_{7,3}+x_{12,4}$ & $c_{24}=x_{4,5}+x_{10,1}+x_{12,4}$  \\
\hline
$c_{12}=x_{2,3}+x_{7,4}+x_{12,5}$ & $c_{25}=x_{5,1}+x_{10,2}+x_{12,5}$ \\
\hline
\end{tabular}
\vspace{5pt}
\caption{Vector linear index code for SUICP-SNI given in Example \ref{ex10}}
\label{table11}
\vspace{-5pt}
\end{table}
\hrule
\begin{table*}[ht]
\centering
\setlength\extrarowheight{0pt}
\begin{tabular}{|c|c|c|c|c|c|c|c|c|c|}
\hline
$x_{k,1}$ &$\tau_{k,1}$&$x_{k,2}$&$\tau_{k,2}$&$x_{k,3}$&$\tau_{k,3}$&$x_{k,4}$&$\tau_{k,4}$&$x_{k,5}$&$\tau_{k,5}$ \\
\hline
$x_{0,1}$ & $c_0$ & $x_{0,2}$&$c_1$&$x_{0,3}$&$c_2$& $x_{0,4}$&$c_3$&$x_{0,5}$&$c_4$ \\
\hline
$x_{1,1}$ & $c_5$ & $x_{1,2}$&$c_6$&$x_{1,3}$&$c_7$ & $x_{1,4}$&$c_8$&$x_{1,5}$&$c_9$ \\
\hline
$x_{2,1}$ & $c_{10}$ & $x_{2,2}$&$c_{11}$&$x_{2,3}$&$c_{12}$ & $x_{2,4}$&$c_{13}$&$x_{2,5}$&$c_{14}$ \\
\hline
$x_{3,1}$ & $c_{15}$ &$x_{3,2}$&$c_{16}$&$x_{3,3}$&$c_{17}$& $x_{3,4}$&$c_{18}$&$x_{3,5}$&$c_{19}$  \\
\hline
$x_{4,1}$ & $c_{20}$ &$x_{4,2}$&$c_{21}$&$x_{4,3}$&$c_{22}$ & $x_{4,4}$&$c_{23}$&$x_{4,5}$&$c_{24}$ \\
\hline
$x_{5,1}$ & $c_{25}$ & $x_{5,2}$&$c_{0}$&$x_{5,3}$&$c_{1}$& $x_{5,4}$&$c_2$&$x_{5,5}$&$c_3$  \\
\hline
$x_{6,1}$ & $c_{4}$ & $x_{6,2}$&$c_{5}$&$x_{6,3}$&$c_{6}$& $x_{6,4}$&$c_7$&$x_{6,5}$&$c_8$ \\
\hline
$x_{7,1}$ & $c_{9}$ & $x_{7,2}$&$c_{10}$&$x_{7,3}$&$c_{11}$ & $x_{7,4}$&$c_{12}$&$x_{7,5}$&$c_{0}+c_{13}$ \\
\hline
$x_{8,1}$& $c_{1}+c_{14}$ & $x_{8,2}$&$c_{2}+c_{15}$&$x_{8,3}$&$c_{3}+c_{16}$& $x_{1,2}$&$c_{4}+c_{17}$&$x_{1,3}$&$c_5+c_{18}$ \\
\hline
$x_{9,1}$ & $c_6+c_{19}$ & $x_{9,2}$&$c_7+c_{20}$&$x_{9,3}$&$c_8+c_{21}$ & $x_{1,2}$&$c_9+c_{22}$&$x_{1,3}$&$c_{10}+c_{23}$ \\
\hline
$x_{10,1}$ & $c_{11}+c_{24}$ & $x_{10,2}$&$c_{12}+c_{25}$&$x_{10,3}$&$c_{13}$& $x_{1,2}$&$c_{14}$&$x_{1,3}$&$c_{15}$  \\
\hline
$x_{11,1}$ & $c_{16}$ & $x_{11,2}$&$c_{17}$&$x_{11,3}$&$c_{18}$ & $x_{1,2}$&$c_{19}$&$x_{1,3}$&$c_{20}$ \\
\hline
$x_{12,1}$ & $c_{21}$ & $x_{12,2}$&$c_{22}$&$x_{12,3}$&$c_{23}$ & $x_{1,2}$&$c_{24}$&$x_{1,3}$&$c_{25}$ \\
\hline
\end{tabular}
\vspace{5pt}
\caption{Decoding of vector linear index code for SUICP-SNI given in Example \ref{ex10}}
\label{table12}
\hrule
\end{table*}
\end{example}
\begin{example}
\label{ex2}
Consider a SUICP-SNI with $K=71,D=25,U=1$. For this SUICP-SNI, we have $(a=1,b=30) \in \mathbf{S}_{71,25,1}$ and  corresponding $D+1+\frac{a}{b}$ is $25.033$. This rate can be achieved by the AIR matrix of size $2130 \times 781$. The encoding matrix for this SUICP-SNI is shown in Fig. \ref{fig10}.
\end{example}
\begin{figure*}[ht]
\centering
\includegraphics[scale=0.65]{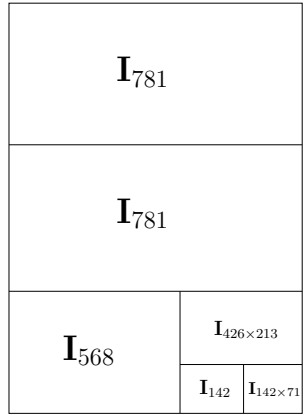}\\
\caption{Encoding matrix for SUICP-SNI given in Example \ref{ex2}}.
\label{fig10}
\end{figure*}

\begin{lemma}
\label{lemma6}
For every  $(K,D,U)$ SUICP-SNI there exists $(a,b) \in \mathcal{\mathbf{S}}$ such that
\begin{align}
\label{gcd84}
D+1+\frac{a}{b} \leq \frac{K}{\left \lfloor \frac{K}{D+1} \right \rfloor}.
\end{align}
\end{lemma}
\begin{proof}
We have 
\begin{align}
\label{gcd83}
\nonumber
\frac{K}{\left\lfloor\frac{K}{D+1}\right\rfloor}&=\frac{\left\lfloor\frac{K}{D+1}\right\rfloor(D+1)+K \text{mod} (D+1)}{\left\lfloor\frac{K}{D+1}\right\rfloor} \\& 
=D+1+\frac{K \text{mod} (D+1)}{\left\lfloor\frac{K}{D+1}\right\rfloor}=D+1+\frac{\alpha}{\gamma},
\end{align}
where $\alpha=K \text{mod} (D+1)$ and $\gamma=\left\lfloor\frac{K}{D+1}\right\rfloor$.

From \eqref{gcd83}, we have
\begin{align}
\label{gcd81}
\alpha=K-\gamma(D+1)
\end{align}
and these values of $\alpha$ and $\gamma$ satisfy the equation $\text{gcd}(K \gamma,\gamma(D+1)+\alpha)\geq \gamma (U+1)$. Hence, $(\alpha,\gamma) \in \mathcal{\mathbf{S}}$. This completes the proof.
\end{proof}
\begin{corollary}
\label{cor2}
For, SUICP-SNI, the vector linear index codes constructed by AIR matrices are within $\frac{K \text{mod} (D+1)}{\left\lfloor\frac{K}{D+1}\right\rfloor}$ symbols per message from the lower bound on broadcast rate given in \eqref{cap4}.
\end{corollary}
\begin{theorem}
\label{thm2}
There exists $(a,b) \in \mathbf{S}_{K,D,U}$ such that the rate $D+1+\frac{a}{b}$ coincide with the results on the exact capacity of SUICP-SNI given in \eqref{cap1},\eqref{cap2} and \eqref{cap3} in the respective settings. 
\end{theorem}
\begin{proof}
\begin{enumerate}
\item []
\item To recover the result corresponding to  \eqref{cap1}:
For a given $U \leq D$, if $K \rightarrow \infty$, then $\frac{K \text{mod} (D+1)}{\left\lfloor\frac{K}{D+1}\right\rfloor} \rightarrow 0$. In Lemma \ref{lemma6}, we proved that 
\begin{align*}
D+1+\frac{a}{b} \leq D+1+\frac{K \text{mod} (D+1)}{\left\lfloor\frac{K}{D+1}\right\rfloor}.
\end{align*}
Hence, $D+1+\frac{a}{b} \rightarrow D+1$. This completes the proof.
\item To recover the capacity corresponding to \eqref{cap2}.
 Substituting $D=1$ in \eqref{gcd84} completes the proof of this case.
\item  To recover  \eqref{cap3} as a special case:
If we take $a=0$ and $b=1$, then we get a setting considered in \eqref{cap3}. Hence, the scalar linear code ($b=1$) considered in \eqref{cap3} is a special case of vector linear codes considered in this paper.
\end{enumerate}
\end{proof}

For a given SUICP-SNI with parameters $K,D$ and $U,$ finding the best pair $(a,b) \in \mathbf{S}_{K,D,U}$ such that $D+1+\frac{a}{b}$ is the minimum is an open problem. However for a given set of values for $K$ and $D$ we have $\mathbf{S}_{K,D,U_1} \supseteq \mathbf{S}_{K,D,U_2}$ for $U_1 < U_2.$ This means that all the rates achievable for $(K,D,U_2)$ are achievable for $(K,D,U_1)$ and possibly smaller rates that are not achievable by $(K,D,U_2).$ This is illustrated in the following example.  

\begin{example}
\label{ex3}
For SUICP-SNI with $K=71$, $U \leq D \leq 15$, one possible pair $(a,b)$ and corresponding $D+1+\frac{a}{b}$ are shown in Table \ref{table1}. For these index coding problems, $D+1$ is given in the $6$th column of Table \ref{table1} gives a lower bound on broadcast rate. The values of $D+1$ can be compared with $D+1+\frac{a}{b}$ in Table \ref{table1}. The rate  $D+1+\frac{a}{b}$ given in $7$th column is achieved by using AIR matrices of size given in the $8$th column of Table \ref{table1}. For SUICP-SNI with $K=71$, the capacity is known to the special case $U=D=1$. For $U=D=1$, $D+1+\frac{a}{b}=2.0285$ coincide with the reciprocal of capacity given in \eqref{cap2}.
\end{example}
\begin{table*}[ht]
\centering
\setlength\extrarowheight{1.5pt}
\begin{tabular}{|c|c|c|c|c|c|c|c|}
\hline
$K$ &$D$&$U$&$a$&$b$&$D+1$&$D+1+\frac{a}{b}$&AIR matrix size \\
\hline
71 & 1 & 1&1&35&2&2.0285&$2485 \times 71$ \\
\hline
71 & 2 & 1,2&2&23&3&3.0869&$1633 \times 71$ \\
\hline
71 & 3 & 1&2&35&4&4.0571&$2485 \times 142$ \\
\hline
71 & 3 & 2,3&5&17&4&4.1764&$1207 \times 71$ \\
\hline
71 & 4 & 1,2,3,4&1&14&5&5.0714&$994 \times 71$ \\
\hline
71 & 5 &1&3&35&6&6.0857&$2485 \times 71$ \\
\hline
71 & 5 & 2&4&23&6&6.1739&$1633 \times 142$ \\
\hline
71 & 5 & 3,4,5&5&11&6&6.4545&$781 \times 71$ \\
\hline
71 & 6 &1,2,$\ldots$,6&1&10&7&7.1000&$710 \times 71$ \\
\hline
71 & 7 &1&4&35&8&8.1142&$2485 \times 284$ \\
\hline
71 & 7 & 2,3&6&17&8&8.3529&$1207 \times142$ \\
\hline
71 & 7 & 4,5,6,7&7&8&8&8.8750&$568 \times 71$ \\
\hline
71 & 8 &1&5&31&9&9.1612&$2201 \times 284$ \\
\hline
71 & 8 & 2&6&23&9&9.2608&$1633 \times 213$ \\
\hline
71 & 8 & 3&7&15&9&9.4666&$1063 \times 142$ \\
\hline
71 & 8 & 4,5,6,7,8&1&7&9&9.1428&$497 \times 71$ \\
\hline
71 & 9 &1,2,$\ldots$,9&1&7&10&10.1428&$497 \times 71$ \\
\hline
71 & 10 & 1&3&32&11&11.0937&$2272 \times 355$ \\
\hline
71 & 10 & 2&4&19&11&11.2105&$1349 \times 213$ \\
\hline
71 & 10 & 3,4,$\ldots$,10&5&6&11&11.8333&$426 \times 71$ \\
\hline
71 & 11 & 1&6&35&12&12.1714&$2485 \times 426$ \\
\hline
71 & 11 & 2&8&23&12&12.3478&$1633 \times 284$ \\
\hline
71 & 11 & 3&9&17&12&12.5294&$1207 \times 213$ \\
\hline
71 & 11 &4,5&10&11&12&12.9090&$781 \times 142$ \\
\hline
71 & 11 & 6,7,$\ldots$,11&11&5&12&14.2000&$355 \times 71$ \\
\hline
71 & 12 &1&4&27&13&13.1481&$1917 \times 355$ \\
\hline
71 & 12 & 2,3&5&16&13&13.3125&$1136 \times 213$ \\
\hline
71 & 12 & 4,5,$\ldots$,12&6&5&13&14.2000&$355 \times 71$ \\
\hline
71 & 13 & 1,2,$\ldots$,13&1&5&14&14.2000&$355 \times 71$ \\
\hline
71 & 14 &1&2&33&15&15.0606&$2343 \times 497$ \\
\hline
71 & 14 &2,3,4&3&14&15&15.2142&$994 \times 213$ \\
\hline
71 & 14 & 5,6&7&9&15&15.7777&$639 \times 142$ \\
\hline
71 & 14 & 7,8,$\ldots$,14&11&4&15&17.7500&$284 \times 71$ \\
\hline
71 & 15 & 1&1&31&16&16.0322&$2201 \times 497$ \\
\hline
71 & 15 & 2&3&22&16&16.1363&$1562 \times 355$ \\
\hline
71 & 15 & 3,4&5&13&16&16.3846&$923 \times 213$ \\
\hline
71 & 15 &5,6,$\ldots$,15&7&4&16&17.7500&$284 \times 71$ \\
\hline
\end{tabular}
\vspace{5pt}
\caption{$D+1+\frac{a}{b}$ for $K=71$ and $U \leq D \leq 15$.}
\label{table1}
\hrule
\end{table*}

\section{Low Complexity Decoding of vector linear index codes for SUICP-SNI}
\label{sec3}

In \cite{VaR4}, we gave a low-complexity decoding for scalar linear index codes for SUICP-SNI with AIR matrix as encoding matrix. The low complexity decoding method helps to identify a reduced set of side-information for each users with which the decoding can be carried out. By this method every receiver is able to decode its wanted message symbol by simply adding some index code symbols (broadcast symbols). The low-complexity decoding is generalized to vector linear index codes of SUICP-SNI in this section.

Consider a $(K,D,U)$ SUICP-SNI. Let $(a,b)\in \mathbf{S}_{K,D,U}$ defined in Definition \ref{def1}. In Theorem \ref{thm1}, we proved that the rate $D+1+\frac{a}{b}$ can be achieved by using the AIR matrix of size $Kb \times (b(D+1)+a)$ as encoding matrix. Let 
\begin{align}
\label{mn}
m=Kb~~~\text{and}~~~n=b(D+1)+a.
\end{align}

The submatrices of an AIR matrix are classified in to the following three types. 
\begin{itemize}
\item The first submatrix is the $I_{n \times n}$ matrix at the top of Fig. \ref{fig1} which is independent of $\lambda_i,\beta_i,$  $i \in [0:l].$ This will be referred as the $I_{n}$ matrix henceforth.
\item The set of matrices of the form $I_{\lambda_i \times \beta_i \lambda_i}$ for $i=0,2,4, \cdots$ (for all $i$ even) will be referred as the set of even-submatrices.
\item The set of matrices of the form $I_{\beta_i \lambda_i \times  \lambda_i}$ for $i=1,3,5, \cdots$ (for all $i$ odd) will be referred as the set of odd-submatrices. 
\end{itemize}
Note that the odd-submatrices are always "fat" and the even-submatrices are always "tall" including square matrices in both the sets. 
By the $i$-th submatrix is meant either an odd-submatrix or an even-submatrix for $ 0 \leq i \leq l.$  Also whenever  $\beta_0=0,$  the corresponding submatrix will not exist in the AIR matrix.  
To prove the main result in the following section the location of both the odd- and even-submatrices within the AIR matrix need to be identified. Towards this end, we define the following intervals. Let $R_0,R_1,R_2,\ldots,R_{\left\lfloor \frac{l}{2}\right\rfloor+1}$ be the intervals that will identify the rows of the submatrices  as given below:

\begin{itemize}
\item $R_0=[0:m-\lambda_0-1]$
\item $R_1=[m-\lambda_0:m-\lambda_2-1]$
\item $R_2=[m-\lambda_2:m-\lambda_4-1]$
\item []~~~~~~~~~~~~~~$\vdots$
\item $R_i=[m-\lambda_{2(i-1)}:m-\lambda_{2i}-1]$
\item []~~~~~~~~~~~~~~$\vdots$
\item $R_{\left\lfloor
\frac{l}{2}\right\rfloor}=[m-\lambda_{2(\left\lfloor
\frac{l}{2}\right\rfloor-1)}:m-\lambda_{2\left\lfloor
\frac{l}{2}\right\rfloor}-1]$
\item $R_{\left\lfloor \frac{l}{2}\right\rfloor+1}=[m-\lambda_{2\left\lfloor \frac{l}{2}\right\rfloor}:m-1]$,
\end{itemize} 
we have $R_0\cup R_1 \cup R_2 \cup \ldots \cup R_{\left\lfloor \frac{l}{2}\right\rfloor +1}=[0:m-1]$.

Let $C_0,C_1,\ldots,C_{\left\lceil \frac{l}{2}\right\rceil}$ be the intervals that will identify the columns of the submatrices  as given below:
\begin{itemize}

\item $C_0=[0:\beta_0\lambda_0-1]$ if $\beta_0 \geq 1$, else $C_0=\phi$
\item $C_1=[n-\lambda_1:n-\lambda_3-1]$
\item $C_2=[n-\lambda_3:n-\lambda_5-1]$
\item [] ~~~~~~~~~~~$\vdots$
\item $C_i=[n-\lambda_{2i-1}:n-\lambda_{2i+1}-1]$
\item []~~~~~~~~~~~~~~$\vdots$
\item $C_{\left\lceil
\frac{l}{2}\right\rceil-1}=[n-\lambda_{2\left\lceil
\frac{l}{2}\right\rceil-3}:n-\lambda_{2\left\lceil
\frac{l}{2}\right\rceil-1}-1]$
\item $C_{\left\lceil \frac{l}{2}\right\rceil}=[n-\lambda_{2\left\lceil \frac{l}{2}\right\rceil-1}:n-1]$
\end{itemize} 
we have $C_0\cup C_1 \cup C_2 \cup \ldots \cup C_{\left\lceil \frac{l}{2}\right\rceil}=[0:n-1]$.

Let $\mathbf{L}$ be the AIR matrix of size $m \times n$. In the matrix $\mathbf{L}$, the element $\mathbf{L}(j,k)$ is present in one of the submatrices: $\mathbf{I}_{n}$ or $\mathbf{I}_{\beta_{2i+1}\lambda_{2i+1} \times \lambda_{2i+1}}$ for $i \in [0:\lceil \frac{l}{2}\rceil-1]$ or $\mathbf{I}_{\lambda_{2i} \times \beta_{2i}\lambda_{2i}}$ for $i \in [0:\left\lfloor\frac{l}{2}\right\rfloor]$. Let $(j_R,k_R)$ be the (row-column) indices of $\mathbf{L}(j,k)$ within the submatrix in which $\mathbf{L}(j,k)$ is present. Then, for a given $\mathbf{L}(j,k)$, the indices $j_R$ and $k_R$ are as given below.
\begin{itemize}
\item If $\mathbf{L}(j,k)$ is present in $\mathbf{I}_{n}$, then $j_R=j$ and $k_R=k$.
\item If $\mathbf{L}(j,k)$ is present in $\mathbf{I}_{\lambda_{0} \times \beta_{0}\lambda_{0}}$, then   $j_R=j~\text{\textit{mod}}~(n)$ and $k_R=k$. 
\item If $L(j,k)$ is present in $\mathbf{I}_{\beta_{2i+1}\lambda_{2i+1} \times \lambda_{2i+1}}$ for $i \in [0:\lceil \frac{l}{2}\rceil-1]$, then \\ $j_R=j~\text{\textit{mod}}~(m-\lambda_{2i})$ and $k_R=k~\text{\textit{mod}}~(n-\lambda_{2i+1})$.
\item If $L(j,k)$ is present in $\mathbf{I}_{\lambda_{2i} \times \beta_{2i}\lambda_{2i}}$ for $i \in [1:\left\lfloor\frac{l}{2}\right\rfloor]$, then $j_R=j~\text{\textit{mod}}~(m-\lambda_{2i})$ and $k_R=k~\text{\textit{mod}}~(n-\lambda_{2i-1})$. 
\end{itemize}

In Definition \ref{defv1} below we define several distances between the $1$s present in an AIR matrix. These distances are used to prove the low complexity decoding for SUICP-SNI.  Figure \ref{sfig44} is useful to visualize the distances defined.

\begin{definition}
\label{defv1}
Let $\mathbf{L}$ be the AIR matrix of size $m \times n$.
\begin{itemize}
\item [\textbf{(i)}] For $k\in [0:n-1]$ we have  $\mathbf{L}(k,k)=1.$  Let $k^{\prime}$ be the maximum integer such that $k^{\prime} > k$ and $\mathbf{L}(k^{\prime},k)=1$. Then $k^{\prime}-k,$ denoted by $d_{down}(k),$  is called the down-distance of $\mathbf{L}(k,k)$.  
\item [\textbf{(ii)}] Let $\mathbf{L}(j,k)=1$ and $j \geq n$. Let $j^{\prime}$ be the maximum integer such that $j^{\prime} < j$ and $\mathbf{L}(j^{\prime},k)=1$. Then $j-j^{\prime},$ denoted by $d_{up}(j,k),$ is called the up-distance of $\mathbf{L}(j,k).$ 

\item [\textbf{(iii)}]  Let $\mathbf{L}(j,k)=1$ and $\mathbf{L}(j,k)\in \mathbf{I}_{ \lambda_{2i} \times \beta_{2i} \lambda_{2i}}$ for $i \in [0:\lfloor \frac{l}{2}\rfloor]$. Let $k^{\prime}$ be the minimum integer such that $k^{\prime} > k$ and $\mathbf{L}(j,k^{\prime})=1$. Then $k^{\prime}-k,$ denoted by $d_{right}(j,k),$  is called the right-distance of $\mathbf{L}(j,k).$  

\item [\textbf{(iv)}] For $k\in [0:n-\lambda_l-1]$, let $d_{right}(k+d_{down}(k),k)=\mu_k.$  Let the number of $1$s in the $(k+\mu_k)$th column of $\mathbf{L}$ below $\mathbf{L}(k+d_{down}(k),k+\mu_k)$ be $p_k$ and these are  at a distance of $t_{k,1},t_{k,2},\ldots,t_{k,p_k}~(t_{k,1}<t_{k,2}<\ldots <t_{k,p_k})$ from $\mathbf{L}(k+d_{down}(k),k+\mu_k)$. Then, $t_{k,\tau}$ is called the $\tau$-th  down-distance of $\mathbf{L}(k+d_{down}(k),k+\mu_k)$, for $1 \leq \tau \leq p_k.$
\end{itemize}
\end{definition}

Notice that if $L(k+d_{down}(k),k+\mu_k) \in \mathbf{I}_{\lambda_{2i} \times   \beta_{2i} \lambda_{2i}}$ in $\mathbf{L}$ for $i \in [0:\lfloor \frac{l}{2}\rfloor]$, then $p_k=0$.

\begin{figure}[ht]
\centering
\includegraphics[scale=0.62]{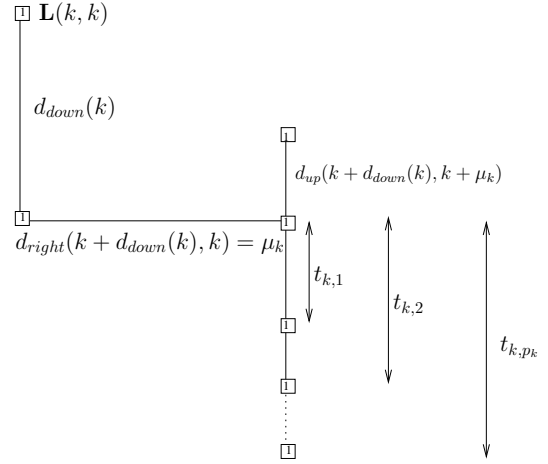}\\
\caption{Illustration of Definition \ref{defv1}}
\label{sfig44}
\end{figure}

\begin{lemma}
\label{lemmav1}
Let $k \in C_i$ for $i \in [0:\lceil\frac{l}{2}\rceil]$. Let $k~\text{mod}~(n-\lambda_{2i-1})=c\lambda_{2i}+d$ for some positive integers $c$ and $d$ $(d<\lambda_{2i})$. The down distance is given by
\begin{align}
\label{mdd}
d_{down}(k)=m-n+\lambda_{2i+1}+(\beta_{2i}-1-c)\lambda_{2i}.
\end{align}
\end{lemma}
\begin{proof}
Proof is given in Appendix A.
\end{proof}

\begin{lemma}
\label{lemmav2}
The up-distance of $\mathbf{L}(j,k)$ is as given below.
\begin{itemize}
\item If $\mathbf{L}(j,k) \in \mathbf{I}_{\beta_{2i+1} \lambda_{2i+1} \times \lambda_{2i+1}}$  for $i \in [0:\lceil \frac{l}{2}\rceil-1]$, then $d_{up}(j,k)$ is $\lambda_{2i+1}$. 
\item If $\mathbf{L}(j,k) \in \mathbf{I}_{\lambda_{2i} \times  \beta_{2i} \lambda_{2i}}$  for $i\in [0:\lfloor \frac{l}{2}\rfloor]$ and $k_R=c\lambda_{2i}+d$ for some positive integer $c,d~(d < \lambda_{2i})$, then $d_{up}(j,k)$ is $\lambda_{2i-1}-c\lambda_{2i}$. 
\end{itemize}
\end{lemma}
\begin{proof}
Proof  is available in Appendix C of \cite{VaR3}.
\end{proof}

\begin{lemma}
\label{lemmav3}
The right-distance of $\mathbf{L}(j,k)$ is as given below.
\begin{itemize}
\item If $k_R \in [0:(\beta_{2i}-1)\lambda_{2i}-1]$ for $i \in [0:\lfloor \frac{l}{2}\rfloor]$, then $d_{right}(j,k)$ is $\lambda_{2i}$. 
\item If $k_R \in [(\beta_{2i}-1)\lambda_{2i}:\beta_{2i} \lambda_{2i}-1]$ for $i \in [0:\lfloor \frac{l}{2}\rfloor-1]$, then $d_{right}(j,k)$ depends on $j_R$. If $j_R=c\lambda_{2i+1}+d$ for some positive integers $c,d~(d<\lambda_{2i+1})$, then  $d_{right}(j,k)$ is $\lambda_{2i}-c\lambda_{2i+1}$. 
\end{itemize}
\end{lemma}
\begin{proof}
Proof  is available in Appendix D of \cite{VaR3}.
\end{proof}

From Euclid algorithm and \eqref{chain}, we can write
\begin{align}
\label{chain2}
\lambda_l=\text{gcd}(m,n).
\end{align}
From \eqref{chain}, we have 
\begin{align}
\label{chain3}
\nonumber
&\lambda_0 > \lambda_1 >\ldots>\lambda_{2i}>\ldots>\lambda_l=\text{gcd}(m,n)~\text{and} \\& 
\lambda_{2i-1}-c\lambda_{2i} \geq \lambda_{2i+1} \geq \lambda_{l} = \text{gcd}(m,n)
\end{align}
for $i \in [0:\lfloor \frac{l}{2}\rfloor]$ and $c \leq \beta_{2i}$.

Define 
\begin{align*}
\tilde{C}_i=\{x+\lambda_0: \forall x \in C_i\}~\text{for}~i\in [0:\left \lceil\frac{l}{2}\right \rceil].
\end{align*}
That is, 
\begin{align}
\label{fact0}
\tilde{C}_i=[m-\lambda_{2i-1}:m-\lambda_{2i+1}-1].
\end{align}
 We have $C_0\cup C_1 \cup \ldots \cup C_{\left\lceil\frac{l}{2}\right\rceil}=[0:n-1]$, hence  
\begin{align*}
\tilde{C}_0 \cup \tilde{C}_1 \cup \ldots \cup \tilde{C}_{\left\lceil\frac{l}{2}\right\rceil}=[\lambda_0:m-1].
\end{align*}

It turns out that the interval $\tilde{C}_{i}$ defined in \eqref{fact0} for $i \in [0:\left\lceil\frac{l}{2}\right\rceil]$
needs to be partitioned into two  as $\tilde{C}_{i} =\tilde{D}_{i}\cup \tilde{E}_{i}$ as given below to prove the main result Theorem \ref{thm1}. Let 
\begin{align}
\label{fact12}
\nonumber
&\tilde{D}_i=[m-\lambda_{2i-1}:m-\lambda_{2i-1}+(\beta_{2i}-1)\lambda_{2i}-1] \\&
\tilde{E}_i=[m-\lambda_{2i-1}+(\beta_{2i}-1)\lambda_{2i}:m-\lambda_{2i+1}-1]. 
\end{align}
for $i \in [0:\left\lceil\frac{l}{2}\right\rceil]$.

A $b$-dimensional vector linear index code of length $n$ generated by an AIR matrix of size $m \times n$ is given by
\begin{align}
\nonumber
 [c_0~c_1~\ldots~c_{n-1}]&=[x_{0,1}~x_{0,2}~\ldots~x_{0,b}~x_{1,1}~\ldots~x_{K-1,b}]\mathbf{L}\\&
 =\sum_{t=0}^{K-1}\sum_{j=1}^{b}x_{t,j}L_{tb+j-1}.
\end{align}
For $k \in [0:m-1]$, let 
\begin{align*}
k_b=\left \lfloor \frac{k}{b} \right \rfloor~\text{and}~k_r=(k~\text{mod}~b)+1,
\end{align*}
i.e., the $k$th row of the matrix $\mathbf{L}$ contains the coefficients used for mixing the message symbol $x_{k_b,k_r}$ in $n$ code symbols.

\begin{theorem}
\label{thmv1}
Let $\mathbf{L}$ be the AIR matrix of size $m \times n$. Let $d_{right}(k+d_{down}(k),k)=\mu_k.$ Also let $t_{k,\tau}$ be the $\tau$-th  down distance of $\mathbf{L}(k+d_{down}(k),k+\mu_k),$ where $1 \leq \tau \leq p_k.$  Let $\tilde{D}_i$ and $\tilde{E}_i$ be the sets as given in \eqref{fact12}. In this SUICP-SNI, receiver $R_t$ wants to decode its $b$ wanted message symbols $x_{t,j}$ for $t \in [0:K-1]$ and $j \in [1:b]$. Let $k=bt+j-1$ and $k^{\prime}=k-\lambda_0$. Then, the receiver $R_t$ can decode its wanted message symbol $x_{t,j}$ as given below.
\begin{itemize}
\item [\textbf{(i)}] If $k \in [0:m-n-1]$, then $R_t$ can decode $x_{t,j}$ from the broadcast symbol $c_{k~\text{mod}~n}$. Receiver $R_t$ can decode $x_{t,j}$ as 
\begin{align}
\label{case4}
 x_{t,j}=c_{k~\text{mod}~n}+\nu_{k~\text{mod}~n},
 \end{align}
where the side-information term $\nu_{k~\text{mod}~n}$ is the exclusive OR of all message symbols present in $c_{k~\text{mod}~n}$ excluding $x_{t,j}$.

\item [\textbf{(ii)}] If $k \in \tilde{D}_i$ for $i \in [0:\lceil \frac{l}{2}\rceil]$, then $R_t$ can decode $x_{t,j}$ by using the two broadcast symbols $c_{k^\prime}, c_{k^\prime+\mu_{k^\prime}}$ and the available side-information. 
That is 
\begin{align}
\label{case1}
x_{t,j}=&c_{k^\prime}+c_{k^\prime+\mu_{k^\prime}}+\underbrace{\nu_{k^\prime}+\nu_{k^\prime+\mu_{k^\prime}}}_{\text{side-information~terms}},
\end{align}
where 
\begin{itemize}
\item $\nu_{k^\prime}$ is the exclusive OR of the set of message symbols present in $c_{k^\prime}$ excluding the message symbols $x_{t,j}$ and $x_{(k^\prime+d_{down}(k^\prime))_b,(k^\prime+d_{down}(k^\prime))_r}$
\item $\nu_{k^\prime+\mu_{k^\prime}}$ is the exclusive OR of the set of message symbols present in $c_{k^\prime+\mu_{k^\prime}}$ excluding the message symbol $x_{(k^\prime+d_{down}(k^\prime))_b,(k^\prime+d_{down}(k^\prime))_r}$. 
\end{itemize}

\item [\textbf{(iii)}] If $k \in \tilde{E}_i$ for $i \in [0:\lceil \frac{l}{2}\rceil-1]$, then $R_t$ can decode $x_{t,j}$ by using the broadcast symbols $c_{k^\prime}, c_{k^\prime}+\mu_{k^\prime},c_{k^{\prime}+t_{k^\prime,1}},\ldots,c_{k^\prime+t_{k^\prime,p_{k^\prime}}}$ and the side-information available. That is 
\begin{align}
\label{case2}
\nonumber
x_{t,j}=&c_{k^\prime}+c_{k^\prime+\mu_{k^\prime}}+c_{{k^\prime}+t_{k^\prime,1}}+\ldots+c_{k^\prime+t_{k^\prime,p_{k^\prime}}}\\&
+\underbrace{\nu_{k^\prime}+\nu_{k^\prime+\mu_{k^\prime}}+\nu_{k^\prime+t_{k^\prime,1}}+\ldots+
\nu_{k^\prime+t_{k^\prime,p_{k^\prime}}}}_{\text{side-information~terms}},
\end{align}
where 
\begin{itemize}
\item $\nu_{k^\prime}$ is the exclusive OR of the set of message symbols present in $c_{k^\prime}$ excluding the message symbols $x_{t,j}$ and $x_{(k^\prime+d_{down}(k^\prime))_b,(k^\prime+d_{down}(k^\prime))_r}$
\item $\nu_{k^\prime+t_{k^\prime,\tau}}$ for $\tau \in [1:p_{k^\prime}]$ is the exclusive OR of the set of message symbols present in $c_{k^\prime}+t_{k^\prime,r}$ excluding the message $x_{{(k^\prime+t_{k^\prime,\tau}+d_{down}(k^\prime))}_b,{(k^\prime+t_{k^\prime,\tau}+d_{down}(k^\prime))}_r}$
\item $\nu_{k^\prime+\mu_{k^\prime}}$ is the exclusive OR of the set of message symbols present in $c_{k^\prime}+\mu_{k^\prime}$ excluding the message symbols $x_{{(k^\prime+d_{down}(k^\prime))}_b,{(k^\prime+d_{down}(k^\prime))}_r}$ and $x_{{(k^\prime+t_{k^\prime,\tau}+d_{down}(k^\prime))}_b,{(k^\prime+t_{k^\prime,\tau}+d_{down}(k^\prime))}_r}$ for $\tau \in [1:p_{k^\prime}]$.
\end{itemize}
\item [\textbf{(iv)}] If $k^\prime \in [m-\lambda_l:m-1]$, then $R_t$ can decode $x_{t,j}$ from the broadcast symbol $c_{k^\prime}$ and the side-information available. That is 
\begin{align}
\label{case3}
x_{t,j}=&c_{k^\prime}+\nu_{k^\prime},
\end{align}
where $\nu_{k^\prime}$ is the exclusive OR of the set of message symbols present in $c_{k^\prime}$ excluding the message symbol $x_{t,j}$.
\end{itemize}
\end{theorem}
\begin{proof}
Proof is given in Appendix B.
\end{proof}
\begin{remark}
\label{rem2}
In $(K,D,U)$ SUICP-SNI, the broadcast symbols used by receiver $R_t$ to decode its wanted message symbols $x_{t,j}$ for $t \in [0:K-1]$ and $j \in [1:b]$ is summarized below. Let $k=bt+j-1$ and $k^{\prime}=k-\lambda_0$. 
\begin{itemize}
\item If $k \in [0:\lambda_0-1]$, then $R_t$ decodes $x_{t,j}$ from $c_{k~\text{mod}~n}$.
\item If $k \in \tilde{D}_i$ for $i \in [0:\left\lceil\frac{l}{2}\right\rceil]$, then $R_t$ decodes $x_{t,j}$ by using two broadcast symbols $c_{k^\prime}$ and $c_{k^\prime+\mu_{k^\prime}}$.
\item If $k \in \tilde{E}_i$ for $i \in [0:\left\lceil\frac{l}{2}\right\rceil-1]$, then $R_k$ decodes $x_{t,j}$ from $c_{k^\prime},c_{k^\prime+\mu_{k^\prime}},c_{k^\prime+t_{k^\prime,r}}$ for $r=1,2,\ldots,p_k$.
\item If $k \in \tilde{E}_i$ for $i=\left\lceil\frac{l}{2}\right\rceil$, then $R_k$ decodes $x_{k,j}$ from $c_{k^\prime}$.
\end{itemize}
\end{remark}

\begin{remark}
\label{rem3}
Let $N_k$ be the number of message symbols present in $c_k$ for $k \in [0:n-1]$. Let $k=bt+j-1$ and $k^{\prime}=k-\lambda_0$. Let the number of side-information used by receiver $R_t$ to decode $x_{t,j}$ be $\gamma_{t,j}$. The value of $\gamma_{t,j}$ is summarized below:
\begin{itemize}
\item If $k \in [0:\lambda_0-1]$, then 
\begin{align*}
\gamma_{t,j}=N_{k~\text{mod}~n}-1.
\end{align*}
\item If $k \in \tilde{D}_i$ for $i \in [0:\left\lceil\frac{l}{2}\right\rceil]$, then 
\begin{align*}
\gamma_{t,j}=N_{k^\prime}+N_{k^\prime+\mu_{k^\prime}}-3,
\end{align*}
where $k^{\prime}=k-\lambda_0$.
\item If $k \in \tilde{E}_i$ for $i \in [0:\left\lceil\frac{l}{2}\right\rceil-1]$, then 
\begin{align*}
\gamma_{t,j}=N_{k^\prime}+N_{k^\prime+\mu_{k^\prime}}+\sum_{j=1}^{p_{k^\prime}} N_{k^\prime+t_{k^\prime,j}}-2p_{k^\prime}-3.
\end{align*}
\item If $k \in \tilde{E}_i$ for $i=\left\lceil\frac{l}{2}\right\rceil$, then  $\gamma_{t,j}=N_{k^\prime}-1$.
\end{itemize}
\end{remark}

\section{Discussion}
\label{sec4}
In this paper, we gave an achievable rate for $(K,D,U)$ SUICP-SNI with arbitrary $K,D$ and $U$ by using AIR matrices. The capacity of $(K,D,U)$ SUICP-SNI is a challenging open problem.

\section*{APPENDIX A}
\subsection*{Proof of Lemma \ref{lemmav1} }
\textit{Case (i)}: $l$ is even and $k \in C_i$ for $i \in [0:\lceil\frac{l}{2}\rceil]$ or  $l$ is odd and $k \in C_i$ for $i \in [0:\lceil\frac{l}{2}\rceil-1]$.

In this case, from the definition of down distance, we have $L(k+d_{down}(k),k) \in \mathbf{I}_{ \lambda_{2i} \times \beta_{2i}\lambda_{2i}}$.
\begin{figure*}[ht]
\centering
\includegraphics[scale=0.67]{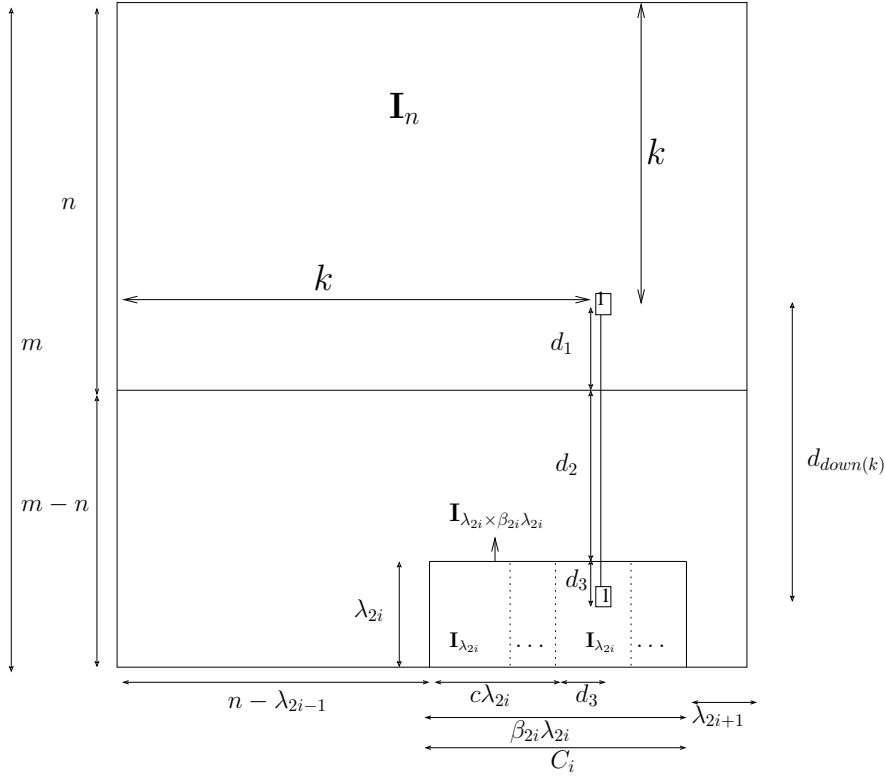}\\
\caption{Maximum-down distance calculation}
\label{afig1}
\end{figure*}
Let $k~\text{mod}~(n-\lambda_{2i-1})=c\lambda_{2i}+d$ for some positive integers $c$ and $d$ $(d<\lambda_{2i})$. 
From Figure \ref{afig1}, we have 
\begin{align}
\label{aeq1}
d_{down}(k)=d_1+d_2+d_3,
\end{align}
and 
\begin{align}
\label{aeq2}
\nonumber
&d_1=n-k, \\&
\nonumber
d_2=m-n-\lambda_{2i}, \\&
d_3=k-(n-\lambda_{2i-1})-c\lambda_{2i}.
\end{align}

By using \eqref{aeq1} and \eqref{aeq2}, we have 
\begin{align}
\label{aeq3}
\nonumber
d_{down}(k)&=d_1+d_2+d_3\\&
\nonumber
=n-k+\underbrace{m-n-\lambda_{2i}}_{d_2}\\&+\underbrace{k-(n-\lambda_{2i-1})-c\lambda_{2i}}_{d_3}\\&
\nonumber
=m-n+\lambda_{2i-1}-(c+1)\lambda_{2i}.
\end{align}

By replacing $\lambda_{2i-1}$ with $\beta_{2i}\lambda_{2i}+\lambda_{2i+1}$ in \eqref{aeq3}, we get 
\begin{align*}
d_{down}(k)=m-n+\lambda_{2i+1}+(\beta_{2i}-1-c)\lambda_{2i}.
\end{align*}
\textit{Case (ii)}: $l$ is odd and $k \in C_{\lceil\frac{l}{2}\rceil}$.

In this case, from the definition of down distance, we have $L(k+d_{down}(k),k) \in \mathbf{I}_{\beta_{l}\lambda_{l} \times \lambda_{l}}$.
\begin{figure*}[ht]
\centering
\includegraphics[scale=0.67]{amatrix}\\
\caption{Maximum-down distance calculation}
\label{afig2}
\end{figure*}
From Figure \ref{afig2}, we have 
\begin{align}
\label{aeq32}
d_{down}(k)=d_1+d_2+d_3,
\end{align}
and
\begin{align}
\label{aeq4}
\nonumber
&d_1=n-k, \\&
\nonumber
d_2=m-n-\beta_l\lambda_{l}, \\&
d_3=\beta_l\lambda_l-d_5.
\end{align}

We have $L(k,k) \in \mathbf{I}_{n}$ and $L(k+d_{down}(k),k) \in \mathbf{I}_{\lambda_l}$ of $\mathbf{I}_{\beta_{l}\lambda_{l} \times \lambda_{l}}$ as shown in Figure \ref{afig2}. Hence, we have $d_1=d_4$ and $d_4=d_5$.
By using \eqref{aeq32} and \eqref{aeq4}, we have 
\begin{align*}
\nonumber
d_{down}(k)&=d_1+d_2+d_3\\&
\nonumber
=d_1+\underbrace{m-n-\beta_l\lambda_{l}}_{d_2}+\underbrace{\beta_l\lambda_l-d_1}_{d_3}=m-n.
\end{align*}

For $i=\left \lceil \frac{l}{2} \right \rceil$, we have $\lambda_{2\left \lceil \frac{i}{2} \right \rceil}=\lambda_{2\left \lceil \frac{i}{2} \right \rceil+1}=0$. We can write $m-n$ as $m-n+\lambda_{2\left \lceil \frac{l}{2} \right \rceil+1}+(\beta_{2\left\lceil\frac{l}{2}\right\rceil}-1-c)\lambda_{2\left \lceil \frac{l}{2} \right \rceil}$. Hence 
\begin{align*}
d_{down}(k)=m-n+\lambda_{2i+1}+(\beta_{2i}-1-c)\lambda_{2i}.
\end{align*}


\section*{APPENDIX B}
It turns out that the interval $\tilde{C}_{i}$ defined in \eqref{fact0} for $i \in [0:\left\lceil\frac{l}{2}\right\rceil]$
needs to be partitioned into two  as $\tilde{C}_{i} =\tilde{D}_{i}\cup \tilde{E}_{i}$ as given below to prove the main result Theorem \ref{thmv1}. Let 
\begin{align}
\label{fact121}
&\tilde{D}_i=[m-\lambda_{2i-1}:m-\lambda_{2i-1}+(\beta_{2i}-1)\lambda_{2i}-1] \\&
\tilde{E}_i=[m-\lambda_{2i-1}+(\beta_{2i}-1)\lambda_{2i}:m-\lambda_{2i+1}-1]. 
\end{align}
for $i \in [0:\left\lceil\frac{l}{2}\right\rceil]$.

\subsection*{Proof of Theorem \ref{thmv1}}
A $b$-dimensional vector linear index code of length $n$ generated by an AIR matrix of size $m \times n$ is given by
\begin{align}
\nonumber
 [c_0~c_1~\ldots~c_{n-1}]&=[x_{0,1}~x_{0,2}~\ldots~x_{0,b}~x_{1,1}~\ldots~x_{K-1,b}]\mathbf{L}\\&
 =\sum_{t=0}^{K-1}\sum_{j=1}^{b}x_{t,j}L_{tb+j-1}.
\end{align}
For $k \in [0:m-1]$, let 
\begin{align*}
k_b=\left \lfloor \frac{k}{b} \right \rfloor~\text{and}~k_r=(k~\text{mod}~b)+1,
\end{align*}
i.e., the $k$th row of the matrix $\mathbf{L}$ contains the coefficients used for mixing the message symbol $x_{k_b,k_r}$ in $n$ code symbols.

Let 
\begin{align*}
k=bt+j-1.
\end{align*}

We prove that for $t \in [0:K-1]$, every receiver $R_{t}$ decodes its wanted message $x_{t,j}$ by using $[c_0~c_1~\ldots~c_{n-1}]$ and its side-information.

{\bf {Case (i)}:} $k \in [0:\lambda_0-1]$

If $m-n < \left\lceil\frac{m}{2}\right \rceil$, the broadcast symbol $c_{k}$ is given by $c_{k}=x_{k_b,k_r}+x_{(k+n)_b,(k+n)_r}$. We have $(k+n)_b=(tb+j-1+b(D+1)+a)_b \geq t+D+1$. Hence, in $c_{k}$, the message symbol $x_{(k+n)_b,(k+n)_r}$ is in the side-information of receiver $R_{t}$. Hence, $R_{t}$ can decode its wanted message symbol $x_{t,j}$ from $c_{k}$. 

If $m-n \geq \left \lceil \frac{m}{2}\right \rceil$, we show that $R_t$ can decode $x_{t,j}$ from $c_{k~\text{mod}~n}$. In this case, from \eqref{chain}, we have $\beta_0=0$ and $\lambda_1=n$. 

If $k \leq n-1$ ($k~\text{mod}~n=k$), From Lemma \ref{lemmav2}, we have 
\begin{align}
\label{fact52}
&d_{up}(k+n,k)=n=b(D+1)+a.
\end{align}

This indicates that the $Db$ message symbols $x_{t+g,h}$ for $g \in [1:D]$,$h \in [1:b]$ and $b-j$ message symbols $x_{t,h}$ for $h \in [j+1:b]$are not present in $c_{k}$. From Lemma \ref{lemmav1}, we have 
\begin{align}
\label{fact521}
\nonumber
d_{down}(k)&=m-n+\lambda_{2i+1}+(\beta_{2i}-1-c)\lambda_{2i}\\&
 \leq m-\lambda_l=m-\text{gcd}(m,n).
\end{align}

This indicates that the $Ub$ message symbols $x_{t-g,h}$ for $g \in [1:U],h \in [1:b]$ and $j-1$ message symbols $x_{t,h}$ for $h \in [1:j-1]$ are not present in $c_{k}$. Hence, every message symbol in $c_{k}$ is in the side-information of $R_t$ excluding the message symbol $x_{t,j}$ and $R_t$ can decode $x_{t,j}$.

If $k \in [n:\lambda_0-1]$, From Lemma \ref{lemmav2}, we have 
\begin{align*}
d_{up}(k,k~\text{mod}~n)&=d_{up}(k+n,k~\text{mod}~n)\\&=n.
\end{align*}

This indicates that the $2Db$ message symbols $x_{t \pm g,h}$ for $g \in [1:D],h \in [1:b]$ and $b-1$ message symbols of $x_{t,h}$ for $h \in [1:b],h \neq j$ are not present in $c_{k~\text{mod}~n}$. Hence, $R_t$ can decode $x_{t,j}$ from $c_{k~\text{mod}~n}$.


{\bf Case (ii):} $k \in \tilde{D_{i}}$  for $i\in [0:\left \lceil \frac{l}{2}\right \rceil]$.

Let $k^{\prime}=k-\lambda_0$. In this case, we have $k^{\prime}_R \in [0:(\beta_{2i}-1)\lambda_{2i}-1]$  for $i\in [0:\left \lceil\frac{l}{2}\right \rceil]$. Let $k^{\prime}_R=c\lambda_{2i}+d$ for some positive integers $c$ and $d$ and $d<\lambda_{2i}$. From Lemma \ref{lemmav3}, we have $\mu_{k^{\prime}}=\lambda_{2i}$, from Definition \ref{defv1}, we have  $t_{k^{\prime},r}=0$ for $r \in [1:p_k]$. From Lemma \ref{lemmav1}, we have 
\begin{align}
\label{fact1}
\nonumber
d_{down}(k^{\prime})&=m-n+\lambda_{2i+1}+(\beta_{2i}-1-c)\lambda_{2i}\\&
=m-n+\lambda_{2i-1}-(c+1)\lambda_{2i}.
\end{align}
From Lemma \ref{lemmav2}, we have 
\begin{align}
\label{fact2}
\nonumber
&d_{up}(k^{\prime}+d_{down}(k^{\prime}),k^{\prime})=\lambda_{2i-1}-c\lambda_{2i}\\& d_{up}(k^{\prime}+d_{down}(k^{\prime}),k^{\prime}+\mu_{k^{\prime}})=\lambda_{2i-1}-(c+1)\lambda_{2i}.
\end{align}
From \eqref{fact1} and \eqref{fact2}
\begin{align}
\label{fact3}
\nonumber
d_{down}(k^{\prime})-d_{up}(k^{\prime}&+d_{down}(k^{\prime}),k^{\prime}+\mu_{k^{\prime}})\\&=m-n=\lambda_0.
\end{align}

This indicates that $x_{t,j}$ is present in the code symbol $c_{k^{\prime}+\mu_{k^{\prime}}}$ and among $(D+1)b+a-1$ message symbols after $x_{t,j}$, only  $x_{(k^{\prime}+d_{down}(k^{\prime}))_b,(k^{\prime}+d_{down}(k^{\prime}))_r}$ is present in $c_{k^{\prime}+\mu_{k^{\prime}}}$. Fig. \ref{sfig31} and \ref{sfig3} illustrate this.
From \eqref{fact2}, 
\begin{align}
\label{fact4}
\nonumber
d_{up}(k^{\prime}&+d_{down}(k^{\prime}),k^{\prime})-d_{up}(k^{\prime}+d_{down}(k^{\prime}),k^{\prime}+\mu_{k^{\prime}})\\&=\lambda_{2i} \geq gcd(m,n) \geq b(U+1).
\end{align}

This along with \eqref{fact3} indicates that every message symbol in $c_{k^{\prime}}$ is in the side-information of $R_{t}$ except $x_{(k^{\prime}+d_{down}(k^{\prime}))_b,(k^{\prime}+d_{down}(k^{\prime}))_r}$. Fig. \ref{sfig31} and \ref{sfig3} illustrate this. Hence, every message symbol in $c_{k^{\prime}}+c_{k^{\prime}+\mu_{k^{\prime}}}$ is in the side-information of $R_{t}$ and $R_{t}$ decodes $x_{t,j}$. 
\begin{figure*}[ht]
\centering
\includegraphics[scale=0.50]{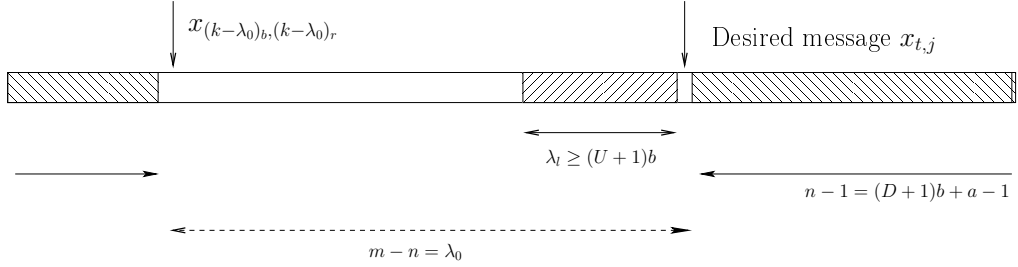}\\
\caption{Decoding for $k=k^{\prime}+\lambda_0 \in \tilde{C}_i$.}
\label{sfig31}
\end{figure*}

\begin{figure}[ht]
\centering
\includegraphics[scale=0.70]{proof8}\\
\caption{Decoding for $k=k^{\prime}+\lambda_0 \in \tilde{D}_i$.}
\label{sfig3}
\end{figure}

{\bf Case (iii):} $k \in \tilde{E_{i}}$  for $i\in [0:\lceil\frac{l}{2}\rceil-1]$.

Let $k^{\prime}=k-\lambda_0$. In this case, we have $k^{\prime}_R \in [(\beta_{2i}-1)\lambda_{2i}:\beta_{2i}\lambda_{2i}-1]$ for $i\in [0:\lceil\frac{l}{2}\rceil-1]$. Let $k^{\prime}_R=(\beta_{2i}-1)\lambda_{2i}+c\lambda_{2i+1}+d$ for some positive integers $c,d \ (d<\lambda_{2i+1})$. We have $k^{\prime}=n-\lambda_{2i-1}+k^{\prime}_R$. From  Lemma \ref{lemmav1}, we have 
\begin{align}
\label{fact5}
d_{down}(k^{\prime})=m-n+\lambda_{2i+1}. 
\end{align}
From Lemma \ref{lemmav3}, we have 
\begin{align}
\label{fact6}
\nonumber
\mu_{k^{\prime}}&=d_{right}(k^{\prime}+d_{down}(k^{\prime}),k^{\prime})\\&
\nonumber
=d_{right}(n-\lambda_{2i-1}+k^{\prime}_R\\&
\nonumber 
~~~~~~~~~~~~+m-n+\lambda_{2i+1},k^{\prime})\\&
\nonumber
=d_{right}(m-\lambda_{2i}+c\lambda_{2i+1}+d,k^{\prime})\\&
=\lambda_{2i}-c\lambda_{2i+1}.
\end{align}
From Lemma \ref{lemmav2}, we have
\begin{align}
\label{fact7}
d_{up}(k^{\prime}+d_{down}(k^{\prime}),k^{\prime}+\mu_{k^{\prime}})=\lambda_{2i+1}.
\end{align}
From \eqref{fact1} and \eqref{fact2}
\begin{align}
\label{fact8}
\nonumber
d_{down}(k^{\prime})-d_{up}(k^{\prime}&+d_{down}(k^{\prime}),k^{\prime}+\mu_{k^{\prime}})\\&=m-n.
\end{align}

This indicates that $x_{t,j}$ is present in the code symbol $c_{k^{\prime}+\mu_{k^{\prime}}}$ and among $(D+1)b+a-1$ messages after $x_{t,j}$, the interfering messages $x_{(k^{\prime}+d_{down}(k^{\prime}))_b,(k^{\prime}+d_{down}(k^{\prime}))_r}$ and $x_{(k^{\prime}+t_{k^{\prime},r}+d_{down}(k^{\prime}))_b,(k^{\prime}+t_{k^{\prime},r}+d_{down}(k^{\prime}))_r}$ for $\tau \in [1:p_{k^{\prime}}]$ are present in $c_{k^{\prime}+\mu_{k^{\prime}}}$. Fig. \ref{sfig4} is useful to understand this.

From Lemma \ref{lemmav1} and Definition \ref{defv1}, $k^{\prime}+d_{down}(k^{\prime})+t_{k^{\prime},p_{k^{\prime}}}$ is always less than the number of rows in the matrix $\mathbf{L}$. That is, $k^{\prime}+t_{k^{\prime},p_{k^\prime}}+d_{down}(k^{\prime})<m$. Hence, we have 
\begin{align}
\label{fact9}
\nonumber
t_{k^{\prime},p_{k^\prime}}&< m-k^{\prime}-d_{down}(k^{\prime})\\&
\nonumber
=m-(n-\lambda_{2i-1}+(\beta_{2i}-1)\lambda_{2i}+c\lambda_{2i+1}+d)-\\&
\nonumber
~~~~~~~~~~~~~~~~~~~~~~~~~~~~~~~~(m-n+\lambda_{2i+1})\\&
=\lambda_{2i}-c\lambda_{2i+1}-d
\end{align}
From \eqref{fact6} and \eqref{fact9}
\begin{align}
\label{relation2}
t_{k^{\prime},p_{k^{\prime}}} < \mu_{k^{\prime}}-d.
\end{align}
From \eqref{fact9}, we have 
\begin{align*}
k^{\prime}_R+t_{k^{\prime},p_{k^{\prime}}}&<k^{\prime}_R+\lambda_{2i}-c\lambda_{2i+1}-d\\&=\underbrace{(\beta_{2i}-1)\lambda_{2i}+c\lambda_{2i+1}+d}_{k^{\prime}_R}+\lambda_{2i}-c\lambda_{2i+1}-d\\&=\beta_{2i}\lambda_{2i}.
\end{align*}
Hence, 
\begin{align*}
k^{\prime}_R+t_{k^{\prime},p_{k^{\prime}}} \in [(\beta_{2i}-1)\lambda_{2i}:\beta_{2i}\lambda_{2i}-1]
\end{align*}
and
\begin{align}
\label{relation3}
\mathbf{L}(k^{\prime}+t_{k^{\prime},r}+d_{down}(k^{\prime}+t_{k^{\prime},r}),k^{\prime}+t_{k^{\prime},r}) \in \mathbf{I}_{\lambda_{2i} \times  \beta_{2i} \lambda_{2i}}
\end{align}
for $r \in [1:p_{k^{\prime}}]$.
We have  
\begin{align}
d_{down}(k^{\prime})=d_{down}(k^{\prime}+t_{k^{\prime},r})
\end{align}
for $r \in [1:p_{k^{\prime}}]$. 

From Lemma \ref{lemmav2}, for $\mathbf{L}(k^{\prime}+t_{k^{\prime},r}+d_{down}(k^{\prime}+t_{k,r}),k^{\prime}+t_{k^{\prime},r})$ for $i\in [0:\lceil\frac{l}{2}\rceil]$,
\begin{align}
\label{fact10}
d_{up}(k^{\prime}+t_{k^{\prime},r}+d_{down}(k^{\prime}+t_{k^{\prime},r}),k^{\prime}+t_{k^{\prime},r})=\lambda_{2i}+\lambda_{2i+1}.
\end{align}
From \eqref{fact7} and \eqref{fact10}, we have 
\begin{align}
\label{fact15}
\nonumber
&d_{up}(k^{\prime}+t_{k^{\prime},r}+d_{down}(k^{\prime}+t_{k,r}),k^{\prime}+t_{k^{\prime},r})-\\&
\nonumber
d_{up}(k^{\prime}+d_{down}(k^{\prime}),k^{\prime}+\mu_{k^{\prime}})=\lambda_{2i}\geq ~gcd(m,n) \\& 
~~~~~~~~~~~~~~~~~~~~~~~~~~~~~~~~~~~~\geq b(U+1).
\end{align}

This along with \eqref{fact8} indicates that every message symbol in $c_{k^{\prime}+t_{k^{\prime},r}}$ is in the side-information of $R_{t}$ except $x_{(k^{\prime}+t_{k^{\prime},r}+d_{down}(k^{\prime}))_b,(k^{\prime}+t_{k^{\prime},r}+d_{down}(k^{\prime}))_r}$. Fig. \ref{sfig4} is useful to understand this. We have $k^{\prime}_R \in [(\beta_{2i}-1)\lambda_{2i}:\beta_{2i}\lambda_{2i}-1]$, $d_{up}(k^{\prime}+d_{down}(k^{\prime}),k^{\prime})=\lambda_{2i}+\lambda_{2i+1}$. From \eqref{fact10}, we have
\begin{align}
\label{fact11}
\nonumber
d_{up}(k^{\prime}&+d_{down}(k^{\prime}),k^{\prime})-d_{up}(k^{\prime}+d_{down}(k^{\prime}),k^{\prime}+\mu_{k^{\prime}})\\&=\lambda_{2i}\geq~gcd(m,n)\geq b(U+1).
\end{align}

This along with \eqref{fact8} indicates that every message symbol in $c_{k^{\prime}}$ is in the side-information of $R_{t}$ except $x_{(k^{\prime}+d_{down}(k^{\prime}))_b,(k^{\prime}+d_{down}(k^{\prime}))_r}$. 

From \eqref{fact8},\eqref{fact15} and \eqref{fact11}, the interfering message symbol $x_{(k^{\prime}+t_{k^{\prime},r}+d_{down}(k^{\prime}))_b,(k^{\prime}+t_{k^{\prime},r}+d_{down}(k^{\prime}))_r}$ in $c_{k^{\prime}+\mu_{k^{\prime}}}$ can be canceled by adding the index code symbol $c_{k^{\prime}+t_{k^{\prime},r}}$  for $r \in [1:p_k]$ and the interfering message symbol $x_{(k^{\prime}+d_{down}(k^{\prime}))_b,(k^{\prime}+d_{down}(k^{\prime}))_r}$ in $c_{k^{\prime}+\mu_{k^{\prime}}}$ can be canceled by adding the index code symbol $c_{k^{\prime}}$. 

Hence, receiver $R_k$ decodes the message symbol $x_k$ by adding the index code symbols $c_{k^{\prime}},c_{k^{\prime}+\mu_{k^{\prime}}}$ and  $c_{k^{\prime}+t_{k^{\prime},r}}$ for $r \in [1:p_{k^{\prime}}]$. 
\begin{figure}[ht]
\centering
\includegraphics[scale=0.60]{proof7}\\
\caption{Decoding for $k=k^{\prime}+\lambda_0 \in \tilde{E}_i$.}
\label{sfig4}
\end{figure}

{\bf Case (iv):} $k \in [m-\lambda_l:m-1]= \tilde{E}_{i}$  for $i=\left\lceil\frac{l}{2}\right\rceil$.
Let $k^{\prime}=k-\lambda_0$.
In this case, from Lemma \ref{lemmav1}, we have 
\begin{align}
\label{fact51}
d_{down}(k^{\prime})=m-n=\lambda_0.
\end{align}

This indicates that $x_{t,j}$ is present in $c_{k^{\prime}}$ and the $Db$ message symbols $x_{t+g,h}$ for $g \in [1:D],h \in [1:b]$ and $b-j$ message symbols $x_{t,h}$ for $h \in [j+1:b]$ are not present in $c_{k^{\prime}}$. From Lemma \ref{lemmav3}, we have 
\begin{align*}
d_{up}(k)=\lambda_l=\text{gcd}(m,n) \geq b(U+1).
\end{align*}

This indicates that the $Ub$ message symbols $x_{t-g,h}$ for $g \in [1:U],h \in [1:b]$ and $j-1$ message symbols $x_{t,h}$ for $h \in [1:j-1]$ are not present in $c_{k^{\prime}}$. Hence, every message symbol in $c_{k^{\prime}}$ is in the side-information of $R_t$ excluding the message symbol $x_{t,j}$ and $R_t$ can decode $x_{t,j}$. From \eqref{fact0} and \eqref{fact12}, case (i), case (ii), case (iii) and case(iv) span $k \in [0:m-1]$. This completes the proof.
\section*{Acknowledgment}
This work was supported partly by the Science and Engineering Research Board (SERB) of Department of Science and Technology (DST), Government of India, through J.C. Bose National Fellowship to B. Sundar Rajan


\begin{thebibliography}{160}
\bibitem{MCJ}
H. Maleki, V. Cadambe, and S. Jafar, ``Index coding – an interference alignment perspective", in IEEE \textit{Trans. Inf. Theory,}, vol. 60, no.9, pp.5402-5432, Sep. 2014.
\bibitem{TIM}
S. A. Jafar, ``Topological interference management through index
coding", IEEE Trans. Inf. Theory, vol. 60, no. 1, pp. 529–568,
Jan. 2014.
\bibitem{ISCO}
Y. Birk and T. Kol, ``Informed-source coding-on-demand (ISCOD) over broadcast channels", in \textit{Proc. IEEE Conf. Comput. Commun.}, San Francisco, CA, 1998, pp. 1257-1264.

\bibitem{YBJK}
Z. Bar-Yossef, Z. Birk, T. S. Jayram and T. Kol, ``Index coding with side-information", in \textit{Proc. 47th Annu. IEEE Symp. Found. Comput. Sci.}, Oct. 2006, pp. 197-206.


\bibitem{OnH}
L Ong and C K Ho, ``Optimal Index Codes for a Class of Multicast Networks with Receiver Side Information”, in \textit{Proc. IEEE ICC}, 2012, pp. 2213-2218.

\bibitem{minrank}
R. Peeters, “Orthogonal representations over finite fields and the chromatic number of graphs”, \textit{Combinatorica}, vol. 16, no. 3, Sept 1996, pp. 417–431.

\bibitem{ICVLP}
A. Blasiak, R. Kleinberg and E. Lubetzky, ``Broadcasting With side-information: Bounding and Approximating the Broadcast Rate", in IEEE \textit{Trans. Inf. Theory,}, vol. 59, no.9, pp.5811-5823, Sep. 2013.

\bibitem{VaR1}
M. B. Vaddi and B. S. Rajan, ``Optimal Vector Linear Index Codes for Some Symmetric Multiple Unicast Problems”, in \textit{Proc. IEEE ISIT}, Barcelona, Spain, July 2016.

\bibitem{VaR2}
M. B. Vaddi and B. S. Rajan, ``Optimal Scalar Linear Index Codes for One-Sided Neighboring Side-Information Problems”, \textit{In Proc. IEEE GLOBECOM Workshop on Network Coding and Applications}, Washington, USA, December 2016.

\bibitem{VaR3}
M. B. Vaddi and B. S. Rajan, ``Low-Complexity Decoding for Symmetric, Neighboring and Consecutive Side-information Index Coding Problems”, in arXiv:1705.03192v2 [cs.IT] 16 May 2017.

\bibitem{VaR4}
M. B. Vaddi and B. S. Rajan, ``Capacity of Some Index Coding Problems with Symmetric Neighboring Interference”, in arXiv: 1705.05060v2 [cs.IT] 18 May 2017.


\end{thebibliography}
\end{document}